\documentclass[11pt]{article}

\usepackage[english]{babel}
\usepackage[utf8]{inputenc}
\usepackage{amsmath}
\usepackage{amssymb}
\usepackage{amsfonts}
\usepackage{amsthm}
\usepackage{mathrsfs} 
\usepackage{mathtools} 
\usepackage{enumerate} 
\usepackage[all]{xy}
\usepackage[pdftex]{graphicx}
\usepackage{color} 
\usepackage{cite}
\usepackage{url}
\usepackage{indent first}
\usepackage[labelfont=bf,labelsep=period,justification=raggedright]{caption}
\usepackage[english]{babel}
\usepackage[utf8]{inputenc}
\usepackage{hyperref}
\usepackage[colorinlistoftodos]{todonotes}
\usepackage[margin=1in]{geometry} 
\usepackage{tikz}

\usepackage{bbm}
 
\mathtoolsset{showonlyrefs}

\newtheorem{theorem}{Theorem}

\newtheorem{lemma}[theorem]{Lemma}
\newtheorem{proposition}[theorem]{Proposition}
\newtheorem{corollary}[theorem]{Corollary}
\newtheorem{conjecture}[theorem]{Conjecture}
\newtheorem{remark}[theorem]{Remark}
\newtheorem{rem}[theorem]{Remark}

\newtheorem{prop}[theorem]{Proposition}

\numberwithin{equation}{section}
\numberwithin{theorem}{section}

\renewcommand{\Pr}{\mathop{\bf Pr\/}}

\newcommand{\E}{\mathop{\bf E\/}}

\newcommand{\Var}{\mathop{\bf Var\/}}

\newcommand{\ab}[1]{\langle#1\rangle}

\newcommand{\floor}[1]{\lfloor {#1} \rfloor}

\newcommand{\ceil}[1]{\lceil {#1}\rceil}

\newcommand{\ind}[1]{^{(#1)}}
\newcommand\defeq{\ensuremath{\stackrel{\rm def}{=}}} 

\DeclareMathOperator{\MC}{\text{Max-Cut}}

\newcommand\Bollobas{Bollob{\'a}s }

\begin{document}

\title{Lower bounds for Max-Cut in $H$-free graphs via semidefinite programming}
\author{Charles Carlson\thanks{Department of Computer Science, University of Colorado Boulder, Boulder, CO 80302. Email: {\tt charles.carlson@colorado.edu}.} , 
Alexandra Kolla\thanks{Department of Computer Science, University of Colorado Boulder, Boulder, CO 80302. Email: {\tt alexandra.kolla@colorado.edu}. Research supported by NSF CAREER grant 1452923 as well as NSF AF grant 1814385} , 
Ray Li\thanks{Department of Computer Science, Stanford University, Stanford, CA 94305. Email: {\tt rayyli@cs.stanford.edu}.  Research supported by an NSF GRF grant DGE-1656518 and by NSF grant CCF-1814629.} , 
Nitya Mani\thanks{Department of Mathematics and Computer Science, Stanford University, Stanford, CA 94305. Email: {\tt nmani@cs.stanford.edu}. Research supported in part by a Stanford Undergraduate Advising and Research Major Grant.} , 
Benny Sudakov\thanks{Department of Mathematics, ETH, 8092 Zurich, Switzerland. benjamin.sudakov@math.ethz.ch.	Research supported in part by SNSF grant 200021-175573.} , 
Luca Trevisan\thanks{Computer Science Division, U.C. Berkeley, Berkeley, CA 94720. Email: {\tt luca@berkeley.edu}. Research supported by the NSF under grant CCF 1815434. Work on this project has also received funding from the European Research Council (ERC) under the European Union's Horizon 2020 research and innovation programme (grant agreement No. 834861).}}
\date{}

\maketitle
\begin{abstract}
For a graph $G$, let $f(G)$ denote the size of the maximum cut in $G$. 
The problem of estimating $f(G)$ as a function of the number of vertices and edges of $G$ has a long history and was extensively studied in the last fifty years. In this paper we propose an approach, based on semidefinite programming (SDP), to prove lower bounds on $f(G)$. We use this approach to find large cuts in graphs with few triangles and in
$K_r$-free graphs.
\end{abstract}

\section{Introduction}

The celebrated Max-Cut problem asks for the largest bipartite subgraph of a graph $G$, i.e.,
for a partition of the vertex set of $G$ into disjoint sets $V_1$ and $V_2$ so that 
the number of edges of $G$ crossing $V_1$ and $V_2$ is maximal. This 
problem has been the subject of extensive research, both from a largely 
algorithmic perspective in computer science and from an extremal perspective in combinatorics.
Throughout, let $G$ denote a graph with $n$ vertices and $m$ edges with maximal cut of size $f(G)$.
The extremal version of Max-Cut problem asks to give bounds on $f(G)$ solely as a function of $m$ and $n$.
This question was first raised more than fifty years ago by Erd\H{o}s~\cite{E1} and has attracted a lot of attention 
since then (see, e.g.,~\cite{E73, ER75, EFPS, AL96, SHE92, BS1, AL05, Su07, CFKS} and their references). 

It is well known that every graph $G$ with $m$ edges has a cut of size at least $m/2$. To see this, consider a random 
partition of vertices of the vertices $G$ into two parts $V_1, V_2$ and estimate the expected number of edges between $V_1$ and $V_2$. On the other hand, already in 1960's Erd\H{o}s~\cite{E1} observed that the constant $1/2$ cannot 
be improved even if we consider very restricted families of graphs, e.g., graphs that contain no short cycles.
Therefore the main question, which has been studied by many researchers, is to estimate the error term $f(G)-m/2$, which we call {\it surplus},  for various families of graphs $G$. 

The elementary bound $f(G) \geq m/2$  was improved by Edwards \cite{E73, E75} who showed that every graph with $m$ edges has a cut of size at least $\frac{m}{2}+ \frac{\sqrt {8m+1}-1}{8}$. 
This result is easily seen to be tight in case $G$
is a complete graph on an odd number of vertices, that is, whenever $m=\binom{k}{2}$ for some odd integer $k$. Estimates on the
second error term for other values of $m$ can be found in~\cite{AH} and~\cite{BS1}.

Although the $\sqrt{m}$ error term is tight in general, it was observed by Erd\H{o}s and Lov\'asz \cite{ER75} that for triangle-free graph it can be improved to at least $m^{2/3+o(1)}$.
This naturally yiels a motivating question: what is the best surplus which can always be achieved if we assume that our family of graphs
is {\it $H$-free}, i.e., no graph contains a fixed graph $H$ as a subgraph.
It is not difficult to show (see, e.g.~\cite{AL03}) that for every fixed
graph $H$ there is some $\epsilon=\epsilon(H)>0$ such that $f(G) \geq \frac{m}{2}+\Omega(m^{1/2+\epsilon})$ for all $H$-free graphs with $m$ edges. However, the problem of 
estimating the error term more precisely is not easy, even for relatively 
simple graphs $H$. It is plausible to conjecture (see~\cite{AL05}) that for every fixed graph $H$ there is a constant $c_H$ such that every $H$-free graph $G$ with $m$ edges has a cut with surplus at least
$\Theta(m^{c_H})$, i.e., there is both a lower bound and an infinite sequence of example showing that exponent $c_H$ can not be improved.
This conjecture is very difficult. Even in the case $H=K_3$ determining the correct error term took almost twenty years. Following the works of~\cite{ER75, PT, SHE92}, Alon~\cite{AL96} proved that every $m$-edge triangle free graph has a cut with surplus of order $m^{4/5}$ and that this is tight up to constant factors. There are several other forbidden graphs $H$ for which we know quite accurately the error term for the extremal Max-Cut problem in $H$-free graphs.  For example, it was proved in~\cite{AL05}, that if $H=C_r$ for $r=4, 6, 10$ then  $c_H=\frac{r+1}{r+2}$. The answer is also known in the case when $H$ is a complete bipartite graph $K_{2,s}$ or $K_{3,s}$ (see~\cite{AL05} for details).

\paragraph{New approach to Max-Cut using semidefinite programming.}
Many extremal results for the Max-Cut problem rely on quite elaborate probabilistic arguments.
A well known example of such an argument is a proof by Shearer~\cite{SHE92} that if $G$ is a triangle-free
graph with $n$ vertices and $m$ edges, and if $d_1, d_2, \ldots, d_n$ are the degrees of its vertices, then
$f(G) \geq \frac{m}{2} + O(\sum_{i=1}^n \sqrt d_i)$. The proof is quite intricate and is based on first choosing a random cut
and then randomly redistributing some of the vertices, depending on how many their neighbors are on the same side as the chosen vertex in the initial cut.
Shearer's arguments were further extended, with more technically involved proofs, in~\cite{AL05} to show that the same lower bound remains valid for graphs $G$ with relatively sparse neighborhoods (i.e., graphs which locally have few triangles).

In this article we propose a different approach to give lower bounds on the $\MC$ of sparse $H$-free graphs using approximation by semidefinite programming (SDP). This approach is intuitive and computationally simple. The main idea was inspired by the celebrated approximation algorithm of 
Goemans and Williamson~\cite{GW95} of the $\MC$: given a graph $G$ with $m$ edges, we first construct an explicit solution for the standard $\MC$ SDP relexation of $G$ which has value at least $(\frac{1}{2}+W)m$ for some positive surplus $W$.
We then apply a Goemans-Williamson randomized rounding, based on the sign of the scalar product with random unit vector, to extract a cut in $G$ whose surplus is within constant factor of $W$.
Using this approach we prove the following result.

\begin{theorem}
	\label{thm:sdp}
	Let $G = (V,E)$ be a graph with $n$ vertices and $m$ edges.
	For every $i \in [n]$, let $V_i$ be any subset of neighbours of vertex $i$ and $\varepsilon_i \leq \frac{1}{\sqrt{\vert V_i \vert}}$.
	Then,
	\begin{align}
	f(G) \geq \frac{m}{2} + \sum_{i=1}^n \frac{\varepsilon_i \vert V_i \vert }{4 \pi} - \sum_{(i,j) \in E} \frac{ \varepsilon_i \varepsilon_j |V_i\cap V_j|  }{2}.
	\end{align}
\end{theorem} 

\noindent
This results implies the Shearer's bound \cite{SHE92}. To see this, set $V_i$ to the neighbors of $i$ and $\varepsilon_i=\frac{1}{\sqrt{d_i}}$ for all $i$. Then, if $G$ is triangle-free graph, then  $|V_i\cap V_j|=0$ for every pair of adjacent vertices $i,j$. 

The fact that we apply Goemans-Williamson SDP rounding in this setting is perhaps surprising for a few reasons.
In general, our result obtains a surplus of $\Omega(W)$ from an SDP solution with surplus $W$, which is not possible in general.
The best cut that can be guaranteed from any kind of rounding of a Max-Cut SDP solution with value $(\frac{1}{2}+W)m$ is $(\frac{1}{2}+\Omega(\frac{W}{\log W}))m$ (see \cite{OW08}). 
Furthermore, this is achieved using the RPR$^2$ rounding algorithm, not the Geomans-Williamson rounding algorithm.
Nevertheless, we show that our explicit Max-Cut solution has additional properties that circumvents these issues and permits a better analysis.

\paragraph{New lower bound for Max-Cut of triangle sparse graphs}

Using Theorem~\ref{thm:sdp}, we give a new result on the Max-Cut of triangle sparse graphs that is more convenient to use than previous similar results.
A graph $G$ is \emph{$d$-degenerate} if there exists an ordering of the vertices $1,\dots,n$ such that vertex $i$ has at most $d$ neighbors $j<i$.
Equivalently, a graph is $d$-degenerate if every induced subgraph has a vertex of degree at most $d$.
Degeneracy is a broader notion of graph sparseness than maximum degree: all maximum degree $d$ graphs are $d$-degenerate, but the star graph is 1-degenerate while having maximum degree $n-1$.
Theorem~\ref{thm:sdp} gives the following useful corollary on the Max-Cut of $d$-degenerate graphs. 
\begin{corollary}
	\label{cor:sdp-2}
	Let $\varepsilon\le \frac{1}{\sqrt{d}}$.
	Let $G$ be a $d$-degenerate graph with $m$ edges and $t$ triangles.
	Then
	\begin{align}
	f(G)\ge \frac{m}{2} + \frac{\varepsilon m}{4\pi} - \frac{\varepsilon^2 t}{2}.
	\end{align}
\end{corollary}

\noindent
Indeed, let $1,\dots,n$ be an ordering of the vertices such that any $i$ has at most $d$ neighbors $j<i$, and let $V_i$ be this set of neighbors.
Let $\varepsilon_i=\varepsilon$ for all $i$.
In this way, $\sum_{i}^{} |V_i|$ counts every edge exactly once and $\sum_{(i,j)\in E}^{} |V_i\cap V_j|$ counts every triangle exactly once, and the result follows.
This shows that graphs with few triangles have cuts with surplus similar to triangle-free graphs. 

This result is new and more convenient to use than existing results in this vein, because it relies only on the global count of the number of triangles, rather than a local triangle sparseness property assumed by prior results. 
For example, it was shown that (using Lemma 3.3 of \cite{AL05}) a $d$-degenerate graph with a local triangle-sparseness property, namely that every large induced subgraph with a common neighbor is sparse, has Max-Cut at least $\frac{m}{2} + \Omega(\frac{m}{\sqrt{d}})$.
However, we can achieve the same result with only the guarantee that the global number of triangles is small.
In particular, when there are at most $O(m\sqrt{d})$ triangles, which is always the case with the local triangle-sparseness assumption above, setting $\varepsilon=\Theta(\frac{1}{\sqrt{d}})$ in Corollary~\ref{cor:sdp-2} gives that the Max-Cut is again at least $\frac{m}{2} + \Omega(\frac{m}{\sqrt{d}})$.

\paragraph{Corollary: Lower bounds for Max-Cut of $H$-free degenerate graphs.}
We illustrate usefulness of the above results by giving the following lower bound on the Max-Cut of $K_r$-free graphs.
\begin{theorem}
	\label{thm:kr-lb}
	Let $r \geq 3$.
	There exists a constant $c = c(r) >0$ such that, for all $K_r$-free $d$-degenerate graphs $G$ with $m$ edges,
	\begin{align}
	f(G) \ge \left(\frac{1}{2} +\frac{c}{d^{1-1/(2r-4)}}\right)m.
	\label{eq:kr-lb}
	\end{align}
\end{theorem} 
\noindent
Lower bounds such as Theorem~\ref{thm:kr-lb} giving a surplus of the form $c\cdot \frac{m}{d^\alpha}$ are more fine-grained than those that depend only on the number of edges.
Accordingly, they are useful for obtaining lower bounds the Max-Cut independent of the degeneracy:
many tight Max-Cut lower bounds in $H$ free graphs of the form $\frac{m}{2}+cm^{\alpha}$ first establish that $f(G) \ge \frac{m}{2} + c\cdot \frac{m}{\sqrt{d}}$ for all $H$-free graphs, and then case-working on the degeneracy. \cite{AL05}

In the case of $r=4$ one can use our arguments together with Alon's result on Max-Cut in triangle-free graphs to improve Theorem~\ref{thm:kr-lb} further to $m/2+cm/d^{2/3}$.
While Theorem~\ref{thm:kr-lb} gives nontrivial bounds for $K_r$-free graphs, we believe that a stronger statement is true and propose the following conjecture.

\begin{conjecture}\label{conj:opt}
For any graph $H$, there exists a constant $c=c(H)>0$ such that, for all $H$-free $d$-degenerate graphs with $m\ge 1$ edges,
\begin{align}
	f(G)\ge \left(\frac{1}{2} + \frac{c}{\sqrt{d}}\right)m.
  \label{eq:conj_opt}
\end{align}
\end{conjecture}

Our Theorem~\ref{thm:sdp} implies this conjecture for various graphs $H$, e.g.,  $K_{2,s}, K_{3,s}, C_{r}$ and for any graph $H$ which contains a vertex whose deletion makes it acyclic. 
This was already observed in~\cite{AL05} using the weaker, locally triangle-sparse form of Corollary~\ref{cor:sdp-2} described earlier. 

Conjecture~\ref{conj:opt} provides a natural route to proving a closely related conjecture proposed by Alon, Bollob{\'a}s, Krivelevich, and Sudakov~\cite{AL03}. 
\begin{conjecture}[\cite{AL03}]
\label{conj:abks}
For any graph $H$, there exists constants $\varepsilon=\varepsilon(H)>0$ and $c=c(H)>0$ such that, for all $H$-free graphs with $m\ge 1$ edges,
\begin{align}
  f(G)\ge \frac{m}{2} + cm^{3/4+\varepsilon}.
\end{align}
\end{conjecture}
Since every graph with $m$ edges is obviously $\sqrt{2m}$-degenerate, the Conjecture \ref{conj:opt} implies immediately a weaker form of Conjecture~\ref{conj:abks} with surplus of order $m^{3/4}$. With some extra technical work (see Appendix~\ref{sec:conj}) we can show that it actually implies the full conjecture, achieving a surplus of $m^{3/4+\varepsilon}$ for any graph $H$.
For many graphs $H$ for which Conjecture~\ref{conj:abks} is known, \eqref{eq:conj_opt} was implicitly established for $H$-free graphs \cite{AL05}, making Conjecture~\ref{conj:opt} a plausible stepping stone to Conjecture~\ref{conj:abks}. 
As further evidence of the plausibility of Conjecture~\ref{conj:opt}, we show that Conjecture~\ref{conj:abks} implies a weaker form of Conjecture~\ref{conj:opt}, namely that any $H$-free graph has Max-Cut $\frac{m}{2} + cm \cdot d^{-5/7}$.
Using similar techniques, we can obtain nontrivial, unconditional results on the Max-Cut of $d$-degenerate $H$-free graphs for particular graphs $H$.
See Appendix~\ref{app:gen} for a table of results and proofs.

Conjecture~\ref{conj:opt}, if true, gives a surplus of $\Omega(\frac{m}{\sqrt{d}})$ that is optimal up to a multiplicative constant factor for every fixed graph $H$ which contains a cycle. To see this, consider an Erd\H{o}s-R\'enyi
random graph $G(n,p)$ with $p=n^{-1+\delta}$. Using standard Chernoff-type estimates, one can easily show that with high probability that this graph is $O(np)$-degenerate and its Max-Cut has size at most $\frac{1}{4}\binom{n}{2}p+O(n\sqrt{np})$. Moreover, if $\delta=\delta(H)>0$ is small enough, then with high probability $G(n,p)$ contains only very few copies of $H$ which can be destroyed by deleting few vertices, without changing the degeneracy and surplus of the Max-Cut (see Appendix~\ref{app:ub}).

\section{Lower bounds for Max-Cut using SDP}
\label{sec:general}
In this section we give a lower bound for $f(G)$ in graphs with few triangles, showing Theorem~\ref{thm:sdp}.
To prove this result, we make heavy use of the SDP relaxation of the $\MC$ problem, formulated below for a graph $G = (V, E)$:
\begin{align}
\text{maximize}\qquad & \sum_{(i,j)\in E}^{} \frac{1}{2}(1-\ab{v\ind{i},v\ind{j}})\nonumber\\
\text{subject to}\qquad & \|v\ind{i}\|^2=1\,\forall i\in V.
\label{eq:sdp}
\end{align}
We leverage the classical Goemans-Williamson \cite{GW95} rounding algorithm which that gives an integral solution from a vector solution to the $\MC$ SDP.

\begin{proof}[Proof of Theorem~\ref{thm:sdp}]
	For $i\in[n]$, define $\tilde v\ind{i}\in\mathbb{R}^n$ by
	\begin{align}
	\tilde v\ind{i}_j \ &= \ \left\{
	\begin{tabular}{ll}
	1 & $i=j$\\
	$-\varepsilon_i$ & $j \in V_i$\\
	0 & otherwise.\\
	\end{tabular}
	\right..
	\end{align}
	For $i\in[n]$, let $v\ind{i}\defeq \frac{\tilde v\ind{i}}{\|\tilde v\ind{i}\|}\in\mathbb{R}^n$.
	Then $1\le \|\tilde v\ind{i}\|\le 1 + \varepsilon_i^2 \vert V_i \vert \le 2$ for all $i$.
	For each edge $(i,j)$ with $i \in V_j$, we have
	\begin{align}
	v_i^{(i)} v_i^{(j)} = \frac{1}{\|\tilde v\ind{i}\|}\cdot \frac{-\varepsilon_j}{\|\tilde v\ind{j}\|} \le \frac{-\varepsilon_j}{4}.
	\end{align} 
	For $k \in V_i \cap V_j$, we have $v_k^{(i)} v_k^{(j)}\le \varepsilon_i \varepsilon_j$.
	For $k\not \in \{i,j\}\cup (V_i \cap V_j)$, we have $v_k^{(i)} v_k^{(j)}=0$ as $v_k\ind{i} = 0$ or $v_k\ind{j} = 0$.
	Thus, for all edges $(i,j)$,
	\begin{align}
	\ab{v\ind{i},v\ind{j}}\le -\frac{\varepsilon_i}{4}\mathbbm{1}_{V_j}(i) -\frac{\varepsilon_j}{4}\mathbbm{1}_{V_i}(j) + |V_i\cap V_j| \varepsilon_i \varepsilon_j.
	\label{}
	\end{align}
	
	Here, $\mathbbm{1}_S(i)$ is 1 if $i\in S$ and 0 otherwise.
	Vectors $v\ind{1},\dots,v\ind{n}$ form a vector solution to the SDP~\eqref{eq:sdp}.
	We now round this solution using the Goemans-Williamson \cite{GW95} rounding algorithm.
	Let $w$ denote a uniformly random unit vector, $A=\{i\in[n]:\ab{v\ind{i},w}\ge 0\}$, and $B=[n]\setminus A$.
	Note that the angle between vectors $v\ind{i},v\ind{j}$ is equal to $\cos^{-1}(\ab{v\ind{i},v\ind{j}})$, so the probability an edge $(i,j)$ is cut is
	\begin{align}
	\Pr[(i,j)\text{ cut}]
	\ &= \  \frac{\cos^{-1}(\ab{v\ind{i},v\ind{j}})}{\pi} \nonumber\\
	\ &= \ \frac{1}{2} - \frac{\sin^{-1}(\ab{v\ind{i},v\ind{j}})}{\pi} \nonumber\\  
	\ &\ge \ \frac{1}{2} - \frac{1}{\pi}\sin^{-1}\left(|V_i\cap V_j| \varepsilon_i \varepsilon_j -\frac{\varepsilon_i}{4}\mathbbm{1}_{V_j}(i) -\frac{\varepsilon_j}{4}\mathbbm{1}_{V_i}(j) \right) \nonumber\\  
	\ &\ge \ \frac{1}{2} - \frac{1}{\pi} \cdot \left(
	\frac{\pi}{2}\cdot |V_i\cap V_j| \varepsilon_i \varepsilon_j -\frac{\varepsilon_i}{4}\mathbbm{1}_{V_j}(i) -\frac{\varepsilon_j}{4}\mathbbm{1}_{V_i}(j)
	\right) \nonumber\\  
	\ &= \ \frac{1}{2} + \frac{\varepsilon_i}{4\pi}\mathbbm{1}_{V_j}(i) + \frac{\varepsilon_j}{4\pi}\mathbbm{1}_{V_i}(j) - \frac{|V_i\cap V_j| \varepsilon_i \varepsilon_j}{2}.\nonumber
	\end{align}
	In the last inequality, we used that, for $a,b\in[0,1]$, we have $\sin^{-1}(a-b)\le \frac{\pi}{2}a - b$. 
	This is true as $\sin^{-1}(x)\le \frac{\pi}{2}x$ when $x$ is positive and $\sin^{-1}(x)\le x$ when $x$ is negative.
	Thus, the expected size of the cut given by $A\sqcup B$ is, by linearity of expectation,
	\begin{align}
	\sum_{(i,j)\in E}^{} \Pr[(i,j)\text{ cut}]
	\ &\ge \ \sum_{\substack{ (i,j) \in E \\ i < j}}
	\left(\frac{1}{2} + \frac{\varepsilon_i}{4}
	\mathbbm{1}_{V_j}(i) 
	+ \frac{\varepsilon_j}{4}
	\mathbbm{1}_{V_i}(j) 
	- \frac{|V_i\cap V_j| \varepsilon_i \varepsilon_j}{2} \right) \nonumber\\
	\ &= \  \frac{m}{2} + \sum_{i=1}^n \frac{ \vert V_i \vert \varepsilon_i }{4 \pi}- \sum_{(i,j) \in E} \frac{ |V_i\cap V_j| \varepsilon_i \varepsilon_j}{2}.
	\qedhere
	\end{align}
\end{proof}

In the proof of Theorem~\ref{thm:kr-lb} we use the following consequence of Corollary \ref{cor:sdp-2}.
\begin{corollary}
	\label{cor:lb-2}
	There exists an absolute constant $c>0$ such that the following holds.
	For all $d\ge 1$ and  $\varepsilon\le\frac{1}{\sqrt{d}}$, if a $d$-degenerate graph $G = (V, E)$ has $m$ edges and at most $\frac{m}{8\varepsilon}$ triangles then
	\begin{align}
	f(G) \ge \left(\frac{1}{2} + c \varepsilon\right)  \cdot m.
	\label{}
	\end{align}
\end{corollary}
 
 \section{Decomposition of degenerate graphs}
 
 In a graph $G=(V,E)$, let $n(G)$ and $m(G)$ denote the number of vertices and edges, respectively.
 For a vertex subset $V'\subset V$, let $G[V']$ denote the subgraph induced by $V'$.
 We show that $d$-degenerate graphs with few triangles have small subsets of neighborhoods with many edges. 
 \begin{lemma}[]
 	\label{lem:lb-3}
 	Let $d\ge 1$ and $\varepsilon>0$, and let $G=(V,E)$ be a $d$-degenerate graph with at least $\frac{m(G)}{\varepsilon}$ triangles. 
 	Then there exists a subset $V'$ of at most $d$ vertices with a common neighbor in $G$ such that the induced subgraph $G[V']$ has at least $\frac{|V'|}{\varepsilon}$ edges.
 \end{lemma}
 \begin{proof}
 	Since $G$ is $d$-degenerate, we fix an ordering $1,\dots,n$ of the vertices such that $d_<(i) \le d$ for all $i\in[n]$, where $d_<(i)$ denotes the number of neighbors $j<i$ of $i$. Then, if $t_<(i)$ denotes the number of triangles $\{i,j,k\}$ of $G$ where $j,k < i$, we have
 	\begin{align}
 	\sum_{i}^{} t_<(i)
 	\ = \ t(G)
 	\ \ge \ \frac{m(G)}{\varepsilon}
 	\ = \ \sum_{i=1}^{n} \frac{d_<(i)}{\varepsilon}. 
 	\label{}
 	\end{align}
 	Hence, there must exist some $i$ such that $t_<(i) \ge \frac{d_<(i)}{\varepsilon}$.
 	Let $V'$ denote the neighbors of $i$ with index less than $i$.
 	By definition, the vertices of $V'$ have common neighbor $i$.
 	Additionally, $G[V']$ has at least $\frac{d_<(i)}{\varepsilon}$ edges and $d_<(i)\le d$ vertices, proving the lemma.
 \end{proof}

 We use this lemma to partition the vertices of any $d$-degenerate graph in a useful way.
 
 \begin{lemma}
 	\label{lem:lb-4}
 	Let $\varepsilon > 0$.
 	Let $G = (V, E)$ be a $d$-degenerate graph on $n$ vertices with $m$ edges.
 	Then there exists a partition $V_1,\dots,V_{k+1}$ of the vertex set $V$ with the following properties.
 	\begin{enumerate}
 		\item For $i=1,\dots,k$, the vertex subset $V_i$ has at most $d$ vertices and has a common neighbor, and the induced subgraph $G[V_i]$ has at least $\frac{|V_i|}{\varepsilon}$ edges. 
 		\item The induced subgraph $G[V_{k+1}]$ has at most $\frac{m(G[V_{k+1}])}{\varepsilon}$ triangles.
 	\end{enumerate}
 	\label{lem:technical}
 \end{lemma}
 \begin{proof}
 	We construct the partition iteratively.
 	Let $V_0^*=V$.
 	For $i\ge 1$, we partition the vertex subset $V_{i-1}^*$ into $V_i\sqcup V_i^*$ as follows.
 	If $G[V_{i-1}^*]$ has at least $\frac{m(G[V_{i-1}^*])}{\varepsilon}$ triangles, then by applying Lemma~\ref{lem:lb-3} to the induced subgraph $G[V_{i-1}^*]$, there exists a vertex subset $V_i$ with a common neighbor in $V_{i-1}^*$ such that $|V_i|\le d$ and the induced subgraph $G[V_i]$ has at most $\frac{|V_i|}{\varepsilon}$ edges.
 	In this case, let $V_i^*\defeq V_{i-1}^*\setminus V_i$.
 	Let $k$ denote the maximum index such that $V_k^*$ is defined, and let $V_{k+1}\defeq V_k^*$.
 	By construction, $V_1,\dots,V_k$ satisfy the desired conditions.
 	By definition of $k$, the induced subgraph $G[V_k^*]$ has at most $\frac{m(G[V_k^*])}{\varepsilon}$ triangles, so for $V_{k+1}=V_k^*$, we obtain the desired result.
 \end{proof}

 \subsection{Large Max-Cut from decompositions}
 \label{s:framework}
 For a $d$-degenerate graph $G=(V,E)$, in a partition $V_1,\dots,V_{k+1}$ of $V$ given by Lemma~\ref{lem:lb-4}, the induced subgraph $G[V_{k+1}]$ has few triangles, and thus, by Corollary~\ref{cor:sdp-2}, has a cut with good surplus. This allows us to obtain the following technical result regarding the Max-Cut of $H$-free $d$-degenerate graphs.

 \begin{lemma} 
 	\label{lem:general} 
 	There exists an absolute constant $c>0$ such that the following holds.
 	Let $H$ be a graph and $H'$ be obtained by deleting any vertex of $H$.
 	Let $0<\varepsilon<\frac{1}{\sqrt{d}}$.
 	For any $H$-free $d$-degenerate graph $G=(V,E)$, one of the following holds:
 	\begin{itemize}
 		\item We have
 		\begin{align}
 		f(G) \ge \left( \frac{1}{2} + c\varepsilon \right)m.
 		\label{e:cond1}
 		\end{align}
 		\item There exist graphs $G_1,\dots,G_k$ such that five conditions hold: (i) graphs $G_i$ are $H'$-free for all $i$, (ii) $n(G_i)\le d$ for all $i$, (iii) $m(G_i) \ge \frac{n(G_i)}{8\varepsilon}$ for all $i$, (iv) $n(G_1)+\cdots+n(G_k)\ge \frac{m}{6d}$, and (v)
 		\begin{align}
 		f(G) \ge \frac{m(G)}{2} + \sum_{i=1}^{k} \left(f(G_i) - \frac{m(G_i)}{2}\right).
 		\label{e:cond2}
 		\end{align}
 	\end{itemize}
 \end{lemma}
 \begin{proof}
 	Let $c_1 < 1$ be the parameter given by Corollary~\ref{cor:lb-2}.
 	Let $c = \frac{c_1}{6}$.
 	Let $G = (V, E)$ be a $d$-degenerate $H$-free graph. 
 	Applying Lemma~\ref{lem:technical} with parameter $8\varepsilon$, we can find a partition $V_1,\dots,V_{k+1}$ of the vertex set $V$ with the following properties.
 	\begin{enumerate}
 		\item For $i=1,\dots,k$, the vertex subset $V_i$ has at most $d$ vertices and has a common neighbor, and the induced subgraph $G[V_i]$ at least $\frac{|V_i|}{8\varepsilon}$ edges. 
 		\item The subgraph $G[V_{k+1}]$ has at most $\frac{m(G[V_{k+1}])}{8\varepsilon}$ triangles. 
 	\end{enumerate}
 	For $i=1,\dots,k+1$, let $G_i\defeq G[V_i]$ and let $m_i\defeq m(G_i)$.
 	For $i=1,\dots,k$, since $G$ is $H$-free and each $V_i$ is a subset of some vertex neighborhood in $G$, the graphs $G_i$ are $H'$-free.
 	For $i = 1, \ldots , k$, fix a maximal cut of $G_i$ with associated vertex partition $V_i = A_i \sqcup B_i$.
 	By the second property above, the graph $G_{k+1}$ has at most $\frac{m_{k+1}}{8\varepsilon}$ triangles. 
 	Applying Corollary~\ref{cor:lb-2} with parameter $\varepsilon$, we can find a cut of $G_{k+1}$ of size at least $(\frac{1}{2} + c_1\varepsilon)m_{k+1}$ with associated vertex partition $V_{k+1} = A_{k+1} \sqcup B_{k+1}$.
 	
 	We now construct a cut of $G$ by randomly combining the cuts obtained above for each $G_i$. 
 	Independently, for each $i=1,\dots,k+1$, we add either $A_i$ or $B_i$ to vertex set $A$, each with probability $\frac12$. 
 	Setting $B=V\setminus A$, gives a cut of $G$.
 	As $V_1,\dots,V_{k+1}$ partition $V$, each of the $m-(m_1+\cdots+m_{k+1})$ edges that is not in one of the induced graphs $G_1,\dots,G_{k+1}$ has exactly one endpoint in each of $A, B$ with probability $1/2$. This allows us to compute the expected size of the cut (a lower bound on $f(G)$ as there is some instantiation of this random process that achieves this expected size).
 	\begin{align}
 	\label{eq:general-1}
 	f(G) \ &\ge \  \frac{1}{2}(m - (m_1+\cdots+m_{k+1})) + \left( \frac{1}{2} + c_1\varepsilon \right) \cdot m_{k+1} + \sum_{i=1}^{k} f(G_i) \nonumber\\
 	\ &= \   \frac{m}{2} + c_1\varepsilon m_{k+1} + \sum_{i=1}^{k} \left( f(G_i) - \frac{m_i}{2} \right).
 	\end{align}
 	We bound~\eqref{eq:general-1} based on the distribution of edges in $G$ in $3$ cases:
 	\begin{itemize}
 		\item $m_{k+1}\ge \frac{m}{6}$.
 		Since $f(G_i)\ge \frac{m_i}{2}$ for all $i=1,\dots,k$, \eqref{e:cond1} holds:
 		\begin{align*}
 		f(G)
 		\ \ge \  \frac{m}{2} + c_1\varepsilon m_{k+1} 
 		\ \ge \  \left(\frac{1}{2} + c\varepsilon\right)\cdot m.
 		\end{align*}
 		\item The number of edges between $V_1\cup\cdots\cup V_{k}$ and $V_{k+1}$ is at least $\frac{2m}{3}$.
 		Then, the cut given by vertex partition $V = A' \sqcup B'$ with $A'=V_1\cup\cdots\cup V_k$ and $B'=V_{k+1}$ has at least $\frac{2m}{3}$ edges, in which case $f(G)\ge \frac{2m}{3}> (\frac{1}{2} + \frac{c_1\varepsilon}{6})\cdot m$, so~\eqref{e:cond1} holds.
 		\item $G' = G[V_1\cup\cdots\cup V_k]$ has at least $\frac{m}{6}$ edges. We show \eqref{e:cond2} holds. By construction, for $i=1,\dots,k$, the graph $G_i$ is $H'$ free, has at most $d$ vertices, and has at least $\frac{m_i}{8\varepsilon}$ edges.
 		Since $G$ is $d$-degenerate, $G'$ is as well, so
 		\begin{align}
 		\frac{m}{6} \le m(G') \le d\cdot n(G') = d\cdot \sum_{i = 1}^k n(G_i),
 		\end{align}
 		Hence $n(G_1)+\cdots+n(G_k)\ge \frac{m}{6d}$.
 		Lastly, by \eqref{eq:general-1}, we have
 		\begin{align*}
 		f(G) \ge \frac{m}{2} + \sum_{i=1}^{k} \left(f(G_i) - \frac{m_i}{2}\right).
 		\end{align*}
 	\end{itemize}
 	This covers all possible cases, and in each case we showed either~\eqref{e:cond1} or~\eqref{e:cond2} holds.
 \end{proof}
 \begin{remark}
 	In Corollary~\ref{cor:lb-2} we can take $c = \frac{1}{60}$, and in Lemma~\ref{lem:general} we can take $c = \frac{1}{360}$.
 \end{remark}

 Lemma~\ref{lem:general} allows us to convert Max-Cut lower bounds on $H$-free graphs to Max-Cut lower bounds on $H$-free $d$-degenerate graphs.
 \begin{lemma}
 	\label{lem:general-2}
 	Let $H$ be a graph and $H'$ be obtained by deleting any vertex of $H$.
 	Suppose that there exists constants $a=a(H')\in[\frac{1}{2}, 1]$ and $c'=c'(H')>0$ such that for all $H'$-free graphs $G$ with $m'\ge 1$ edges, $f(G) \ge \frac{m'}{2} + c'\cdot (m')^{a}$.
 	Then there exists a constant $c=c(H)>0$ such that for all $H$-free $d$-degenerate graphs $G$ with $m\ge 1$ edges,
 	\begin{align*}
 	f(G) \ge \left(\frac{1}{2} + cd^{-\frac{2-a}{1+a}}\right)\cdot m.
 	\end{align*}
 \end{lemma}
 \begin{proof}
 	Let $c_2$ be the parameter in Lemma~\ref{lem:general}.
 	We may assume without loss of generality that $c'\le 1$.
 	Let $G$ be a $d$-degenerate $H$-free graph.
 	Let $\varepsilon\defeq c'd^{-\frac{2-a}{1+a}} < d^{-1/2}$ and $c\defeq \min(c'c_2, \frac{c'}{48})$.
 	
 	Applying Lemma~\ref{lem:general} with parameter $\varepsilon$, either~\eqref{e:cond1} or~\eqref{e:cond2} holds.
 	If~\eqref{e:cond1} holds, then, as desired, $$f(G)\ge \left(\frac{1}{2} + c_2\varepsilon \right)m \ge \left(\frac{1}{2} + cd^{-\frac{2-a}{1+a}} \right)m.$$
 	Else~\eqref{e:cond2} holds. Let $G_1,\dots,G_{k+1}$ be the $H'$-free induced subgraphs satisfying the properties in Lemma~\ref{lem:general}, so that
 	\begin{align*}
 	f(G) 
 	\ &\ge \  \frac{m}{2} + \sum_{i=1}^{k} \left(f(G_i) - \frac{m(G_i)}{2}\right)  \nonumber\\
 	\ &\ge \   \frac{m}{2} + \sum_{i=1}^{k} c'\cdot m(G_i)^{a}.
 	\end{align*}
 	For all $i$, we have
 	\begin{align*}
 	c' \cdot m(G_i)^{a} \overset{(*)}\ge \frac{c'\varepsilon}{8\varepsilon^{1+a}}\cdot n(G_i)^{a}  \overset{(**)}\ge \frac{\varepsilon d}{8(c')^{a}}\cdot n(G_i)  \overset{(+)}\ge \frac{\varepsilon d}{8}\cdot n(G_i),
 	\end{align*}
 	where $(*)$ follows since $m(G_i)\ge \frac{n(G_i)}{8\varepsilon}$, $(**)$ follows since $n(G_i)^{a-1}\ge d^{a-1}$ and $\varepsilon^{1+a} = (c')^{1+a}d^{a-2}$, and $(+)$ follows since $c' \le 1$.
 	Hence, as $n(G_1)+\cdots+n(G_k)\ge \frac{m}{6d}$, we have
 	\begin{align*}
 	f(G) 
 	\ &\ge \   \frac{m}{2} + \varepsilon d \sum_{i=1}^{k}\frac{n(G_i)}{8}
 	\ \ge \   \frac{m}{2} + \frac{\varepsilon m}{48} 
 	\ \ge \ \left(\frac{1}{2} + cd^{-\frac{2-a}{1+a}}\right)\cdot m,
 	\end{align*}
 	as desired.
 \end{proof}

 \section{Max-Cut in $K_r$-free graphs}
 \label{sec:kr}
 In this section we specialize Lemmas~\ref{lem:general} and~\ref{lem:general-2} to the case $H = K_r$ to 
 prove Theorem~\ref{thm:kr-lb}.
 Let $\chi(G)$ denote the chromatic number of a graph $G$, the minimum number of colors needed to properly color the vertices of the graph so that no two adjacent vertices receive the same color. We first obtain a nontrivial upper bound on the chromatic number of a $K_r$-free graph $G$, giving an lower bound (Lemma~\ref{lem:kr-d}) on the Max-Cut of $K_r$-free graphs.
 This lower bound was implicit in \cite{AL03}, but we provide a proof for completeness. The lower bound on the Max-Cut of general $K_r$-free graphs enables us to apply Lemma~\ref{lem:general} to give a lower bound on the Max-Cut of $d$-degenerate $K_r$-free graphs per Theorem~\ref{thm:kr-lb}.
 The following well known lemma gives a lower bound on the Max-Cut using the chromatic number.
 \begin{lemma}(see e.g. Lemma 2.1 of \cite{AL03})\label{lem:mcchi}
 	Given a graph $G = (V, E)$ with $m$ edges and chromatic number $\chi(G) \le t$, we have $f(G) \ge (\frac12 + \frac{1}{2t}) m$.
 \end{lemma}
 \begin{proof}
 	Since $\chi(G) \le t$, we can decompose $V$ into independent subsets $V = V_1,\dots, V_t$. 
 	Partition the subsets randomly into two parts containing $\floor{\frac{t}{2}}$ and $\ceil{\frac{t}{2}}$ subsets $V_i$, respectively, to obtain a cut. 
 	The probability any edge is cut is $\frac{\floor{t/2}\cdot \ceil{t/2}}{\binom{t}{2}}\ge \frac{t+1}{2t}$, so the result follows from linearity of expectation.
 \end{proof}
 
 \begin{lemma}\label{lem:chikr}
 	Let $r\ge 3$ and $G = (V, E)$ be a $K_r$-free graph on $n$ vertices. Then, $$\chi(G) \le 4n^{(r-2)/(r-1)}.$$
 \end{lemma}
 \begin{proof}
 	We proceed by induction on $n$. 
 	For $n\le 4^{r-1}$, the statement is trivial as the chromatic number is always at most the number of vertices.
 	Now assume $G=(V,E)$ has $n>4^{r-1}$ vertices and that $\chi(G)\le 4n_0^{(r-2)/(r-1)}$ for all $K_r$-free graphs on $n_0\le n-1$ vertices.
 	The off-diagonal Ramsey number $R(r,s)$ satisfies $R(r,s) \le \binom{r+s-2}{s-1} \le s^{r-1}$ \cite{ES35}.
 	Hence, $G$ has an independent set $I$ of size $s=\floor{n^{1/(r-1)}}$.
 	The induced subgraph $G[V\setminus I]$ is $K_r$-free and has fewer than $n$ vertices, so its chromatic number is at most $4(n-s)^{(r-2)/(r-1)}$.
 	Hence, $G$ has chromatic number at most
 	\begin{align}
 	1 + 4(n-s)^{(r-2)/(r-1)}
 	\ &= \   1 + 4n^{(r-2)/(r-1)}\left( 1-\frac{s}{n} \right)^{(r-2)/(r-1)} \nonumber\\
 	\ &\overset{(*)}\le \   1 + 4n^{(r-2)/(r-1)} - 4n^{(r-2)/(r-1)} \cdot \frac{s}{3n}
 	\ \overset{(**)}< \   4n^{(r-2)/(r-1)}
 	\label{}
 	\end{align}
 	In $(*)$, we used that $\frac{r-2}{r-1}\ge \frac{1}{2}$, that $\frac{s}{n} \le \frac{1}{4}$, and that $(1-x)^a\le 1 - \frac{x}{3}$ for $a\ge \frac{1}{2}$ and $x\le \frac{1}{4}$.
 	In $(**)$, we used that $s\ge 4$ and hence $\frac{3s}{4} < n^{1/(r-1)}$.
 	This completes the induction, completing the proof.
 \end{proof}
 
 \begin{rem}
 	The upper bound on the off-diagonal Ramsey number $R(r, k^{1/(r-1)})$ has an extra logarithmic factor which suggests that the upper bound on $\chi(G)$ of Lemma~\ref{lem:chikr} can be improved by a logarithmic factor with a more careful analysis.
 \end{rem}
 
 \begin{lemma} 
 	\label{lem:kr-d}
 	If $G$ is a $K_r$-free graph with at most $n$ vertices and $m$ edges, then 
 	\begin{align*}
 	f(G) \ge \left( \frac{1}{2} + \frac{1}{8n^{(r-2)/(r-1)}} \right) m
 	\end{align*}
 \end{lemma}
 \begin{proof}
 	This follows immediately via Lemma~\ref{lem:mcchi}
 	and Lemma~\ref{lem:chikr}. 
 \end{proof}
 
 The above bounds allow us to prove Theorem~\ref{thm:kr-lb}.
 \begin{proof}[Proof of Theorem~\ref{thm:kr-lb}] 
 	Let $G$ be a $d$-degenerate $K_r$-free graph and $\varepsilon=d^{-1 + \frac{1}{2r-4}}$.
 	Let $c_2$ be the parameter given by Lemma~\ref{lem:general}.
 	Let $c = \min(c_2,\frac{1}{388})$.
 	
 	Applying Lemma~\ref{lem:general} with parameter $\varepsilon$, one of two properties hold.
 	If~\eqref{e:cond1} holds, then 
 	\begin{align}
 	f(G)
 	\ \ge \ \left(\frac{1}{2} + c_2\varepsilon\right)m 
 	\ &\ge \  \left(\frac{1}{2} + c d^{- 1 + \frac{1}{2r-4}}\right)m 
 	\end{align}
 	as desired.
 	If~\eqref{e:cond2} holds, there exist graphs $G_1,\dots,G_k$ that are $K_{r-1}$-free with at most $d$ vertices such that $G_i$ has at least $\frac{n(G_i)}{8\varepsilon}$ edges, $n(G_1)+\cdots+n(G_k)\ge \frac{m}{6d}$, and 
 	\begin{align*}
 	f(G) 
 	\ &\ge \  \frac{m}{2} + \sum_{i=1}^{k} \left(f(G_i) - \frac{m(G_i)}{2}\right).
 	\end{align*}
 	For all $i$, we have
 	\begin{align*} 
 	f(G_i) - \frac{m(G_i)}{2}
 	\ &\ge \  \frac{m(G_i)}{8n(G_i)^{(r-3)/(r-2)}}  \nonumber\\
 	\ &\ge \ \frac{n(G_i)}{64\varepsilon n(G_i)^{(r-3)/(r-2)}} 
 	\ \ge \ \frac{n(G_i)}{64\varepsilon d^{(r-3)/(r-2)}} 
 	\ = \ \frac{\varepsilon d n(G_i)}{64}.
 	\end{align*}
 	In the first inequality, we used Lemma~\ref{lem:kr-d}.
 	In the second inequality, we used that $m(G_i)\ge \frac{n(G_i)}{8\varepsilon}$.
 	In the third inequality, we used that $n(G_i)\le d$.
 	Hence, as $d(n(G_1)+\cdots+n(G_k))\ge \frac{m}{6}$, we have as desired that
 	\begin{align}
 	f(G) 
 	\ &\ge \   \frac{m}{2} + \sum_{i=1}^{k}\frac{\varepsilon d n(G_i)}{64}
 	\ \ge \   \frac{m}{2} + \frac{\varepsilon m}{388} 
 	\ \ge \ \left(\frac{1}{2} + cd^{-1 + \frac{1}{2r-4}}\right)\cdot m.
 	\qedhere
 	\end{align}
 \end{proof}
 
\begin{rem}
As we already mentioned in the introduction, we can improve the result of Theorem~\ref{thm:kr-lb} in the case that $r = 4$ using Lemma~\ref{lem:general-2} as follows.
Let $H=K_4$, and $H'=K_3$.
By a result of \cite{AL96}, there exists a constant $c'>0$ such that, for all triangle-free graphs $G$ with $m'\ge 1$ edges, we have $f(G)\ge \frac{m'}{2} + c'(m')^{4/5}$.
By Lemma~\ref{lem:general-2} with $H$ and $H'$ and $a=4/5$, there exists a constant $c>0$ such that any $K_4$-free $d$-degenerate graph $G$ with $m\ge 1$ edges satisfies
\begin{align}
f(G)  \ge    \left( \frac{1}{2} + cd^{-\frac{2-(4/5)}{1+(4/5)}} \right)\cdot m =  \left( \frac{1}{2} + cd^{-2/3} \right)\cdot m.
\end{align}
\end{rem}

\section{Concluding Remarks}
\label{sec:conclusion}
In this paper we presented an approach, based on semidefinite programming (SDP), to prove lower bounds on Max-Cut and used it to find large cuts in graphs with few triangles and in
$K_r$-free graphs. A closely related problem of interest is bounding the \textsf{Max-$t$-Cut} of a graph, i.e. the largest $t$-colorable ($t$-partite) subgraph of a given graph. Our results imply good lower bounds for this problem as well. Indeed, by taking a cut for a graph $G$ with $m$ edges and surplus $W$, one can produce a $t$-cut for $G$ of size 
$\frac{t-1}{t}m +\Omega(W)$ as follows.  Let $A, B$ be the two parts of the original cut. If $t=2s$ is even, simply split randomly both $A, B$ into $s$ parts. If $t=2s+1$ is odd, then put every vertex of $A$ randomly in  the parts $1, \ldots, s$ with probability $2/(2s+1)$ and in the part $2s+1$ with probability $1/(2s+1)$. Similarly, put every vertex of $B$ randomly in  the parts $s+1, \ldots, 2s$ with probability $2/(2s+1)$ and in the part $2s+1$ with probability $1/(2s+1)$.
An easy computation (which we omit here) shows that the expected size of the resulting $t$-cut is
$\frac{t-1}{t}m +\Omega(W)$.

The main open question left by our work is Conjecture~\ref{conj:opt}. Proving this conjecture will require some major new ideas. Even showing that any $d$-degenerate $H$-free graph with $m$ edges has a cut with surplus at least
$m/d^{1-\delta}$ for some fixed $\delta$ (independent of $H$) is out of reach of current techniques.

\vspace{0.3cm}

\noindent
{\bf Acknowledgements.} \, The authors thank Jacob Fox and Matthew Kwan for helpful discussions and feedback. 
The authors thank Joshua Brakensiek for pointing out an error in an earlier draft of this paper.
The authors thank Joshua Brakensiek and Yuval Wigderson for helpful comments on an earlier draft of the paper.

\appendix

\section{Max cut in $H$-free graphs for other $H$}
\label{app:gen}

Our methods above and bound on the $\MC$ of a graph in terms of a global triangle count enable us to give improved bounds on the $\MC$ of $d$-regular or $d$-degenerate $H$-free graphs for a variety of other graphs $H$. These bounds (along with a comparison to existing literature) are summarized in Table~\ref{t:lb}.

Although Theorem~\ref{thm:kr-lb} gives a lower bound on $f(G)$ for $K_r$-free graphs, we can improve on this bound using an ad-hoc approach when $r = 4$:
\begin{prop}
  \label{thm:k4-lb}
  There exists a constant $c>0$ such that for all $K_4$-free $d$-degenerate graphs $G$ with $m\ge 1$ edges,
  \begin{align}
    f(G)\ge \left(\frac{1}{2} + cd^{-2/3}\right)m.
  \end{align}
\end{prop}

Our methods above also allow us to leverage the bounds in \cite{AL05} to give nontrivial lower bounds on the $\MC$ of $d$-degenerate $H$-free graphs for several families of forbidden subgraphs $H$.

\begin{prop}\label{thm:generalh-lb}
For a graph $H$, let
$$\alpha_H(d) =
\begin{cases}
d^{-(r+1)/(2r - 1)} & \text{if }H = W_{r}\text{ and } r \text{ is odd} \\
d^{-7/11} & \text{if } H = K_{3, s} \\
d^{-2/3} &  \text{if }H = K_{4, s} \\
d^{-1/2} & \text{if deleting some vertex from $H$ gives a forest (forest+1)} \\
d^{-2/3} & \text{if deleting two vertices from $H$ gives a forest (forest+2)} \\
\end{cases}.
$$
When $H$ is one of the above, there exists $c=c(H) >0$ such that, for all $H$-free $d$-degenerate graphs $G$ on $m\ge 1$ edges,
$$f(G) \ge \left( \frac12 + c\cdot \alpha_H(d) \right) m.$$
\end{prop}

\begin{proof}
We repeatedly apply Lemma~\ref{lem:general-2} by combining it with results from \cite{AL05}.
Table~\ref{t:prooflb} shows the choices of $H,H',$ and $a$ used in the applications of Lemma~\ref{lem:general-2}, along with the associated surplus bounds on $H'$-free graphs from \cite{AL05} and the resulting surplus bounds on $H$-free $d$-degenerate graphs.
\begin{table}[h!]
\begin{center}
\begin{tabular}{c |c |c |c |c |c}
$H$ & $H'$ & $H'$-free surplus \cite{AL05} & $a$ & $\frac{2-a}{1+a}$  & $d$-deg. $H$-free surplus\\
\hline
\hline 
forest+1 & forest & $c'm$ & 1 & $\frac12$  & $cd^{-1/2}m$\\
forest+2 & forest+1 & $c'm^{4/5}$ & $\frac45$ & $\frac23$  & $cd^{-2/3}m$\\
$W_r$ ($r$ odd) & $C_{r-1}$ & $c'm^{r/(r+1)}$ & $\frac{r}{r+1}$ & $\frac{r+2}{2r+1}$  & $cd^{-(r+2)/(2r+1)}m$\\
$K_{3,s}$ & $K_{2,s}$ & $c'm^{5/6}$ & $\frac56$ & $\frac{7}{11}$  & $cd^{-7/11}m$\\
$K_{4,s}$ & $K_{3,s}$ & $c'm^{4/5}$ & $\frac45$ & $\frac23$  & $cd^{-2/3}m$.
\end{tabular}
\end{center}
\caption{We apply Lemma~\ref{lem:general-2} to the above given $H$ using the listed values of $H'$ and $a$ to obtain the given lower bound.} \label{t:prooflb}
\end{table}

Here, forest+1 means that $H$ is some forbidden subgraph such that removing one vertex from $H$ gives a forest, and forest+2 means that removing two vertices from $H$ gives a forest.
\end{proof}
We note that, for $H=K_{3,s}$, the surplus bound in Proposition~\ref{thm:generalh-lb} can be improved to $cd^{-1/2}$, using similar ideas from \cite{AL05}. We sketch a proof here for completeness.
\begin{proposition}
\label{prop:k3s}
  For all $s\ge 1$, there exists a constant $c>0$ such that for all $K_{3,s}$-free $d$-degenerate graphs $G$ with $m\ge 1$ edges,
  \begin{align}
    f(G) \ge \left( \frac{1}{2} + \frac{c}{\sqrt{d}} \right)m
  \end{align}
\end{proposition}
\begin{proof}
  Let $G$ be a $d$-degenerate, $K_{3,s}$-free graph with associated vertex ordering $1,\dots,n$ and let $N_i$ denote the neighbors $j<i$ of $i$, so that $|N_i|\le d$ for all $i$.
  Since $G$ is $K_{3,s}$-free, the subgraph induced by $N_i$ is $K_{2,s}$-free, so by the K\H{o}v\'ari-S\'os-Tur\'an theorem \cite{KST}, this subgraph has at most $c'\cdot |N_i|^{3/2}\le c'\sqrt{d}\cdot |N_i|$ edges for an absolute $c'>0$, so the total number of triangles $t$ satisfies $t\le c'\sqrt{d}\cdot m$.
  Fix $p>0$ to be chosen later and include each vertex independently with probability $p$ in a set $V'$.
  If the induced subgraph $G[V']$ has $m'$ edges and $t'$ triangles, we have, by Corollary~\ref{cor:sdp-2} with $\varepsilon=\frac{1}{\sqrt{d}}$ that
  \begin{align}
    \E\left[f(G[V'])-\frac{m'}{2}\right]\ge \E\left[\frac{\varepsilon m'}{4\pi} - \frac{\varepsilon^2 t'}{2}  \right]
    =  \frac{p^2 m}{4\pi\sqrt{d}} - \frac{p^3t}{2d}
    \overset{(*)} \ge \frac{m(p^2-2\pi c'p^3)}{4\pi\sqrt{d}},
  \end{align}
  where in $(*)$ we used $t\le c'\sqrt{d}m$.
  Choosing $p=\frac{1}{10c'}$, we have that $\E[f(G[V'])-\frac{m'}{2}]\ge c\frac{m}{\sqrt{d}}$ for some absolute constant $c>0$, so there exists some choice of $V'$ for which $f(G[V'])\ge \frac{m'}{2} + \frac{cm}{\sqrt{d}}$.
  By adding the vertices outside $V'$ to this cut randomly, we conclude $\E[f(G)] \ge \frac{m}{2} + \frac{cm}{\sqrt{d}}$.
\end{proof}

Further, as discussed in the introduction, assuming Conjecture~\ref{conj:abks}, we can get a surplus for $H$-free graphs bounded away from the trivial surplus $\Omega(\frac{m}{d})$ for all $H$.
\begin{prop}
  \label{thm:conditional-kr-lb}
  Assuming Conjecture~\ref{conj:abks}, for any graph $H$, there exist constants $\varepsilon=\varepsilon(H)>0$ and $c=c(H)>0$ such that, for all $H$-free $d$-degenerate graphs $G$ on $m\ge 1$ edges, we have
  \begin{align}
    f(G) \ge \left(\frac{1}{2} + cd^{-5/7 + \varepsilon}\right)m.
  \end{align}
\end{prop}

\begin{proof}
  Fix a graph $H$, and let $H'$ be obtained by deleting any vertex of $H$.
  Assuming Conjecture~\ref{conj:abks}, there exist constants $c'=c'(H'), \varepsilon'=\varepsilon'(H')>0$ such that any $H'$-free graph with $m$ edges satisfies $f(G)\ge \frac{m}{2}+c'm^{3/4+\varepsilon'}$.
  By Lemma~\ref{lem:general-2} with $H$ and $H'$ and $a=3/4+\varepsilon'$, there exists constants $c_3=c_3(H),\varepsilon=\varepsilon(H) >0$ such that any $H$-free $d$-degenerate graph $G$ with $m$ edges satisfies
  \begin{align}
    f(G)
    \ \ge \  \left(\frac{1}{2} + c_3d^{-\frac{5/4 - \varepsilon'}{7/4 + \varepsilon'}}\right)m
    \ \ge \ \left( \frac{1}{2} + c_3d^{-5/7 + \varepsilon} \right)m.
    &\qedhere
  \end{align}
\end{proof}

\begin{table}[ht!]
\begin{center}
\begin{tabular}{l|lr|lr|c}
Forbidden subgraph $H$ & Prior work &  & This work &  & Tight? \\
\hline
\hline
None & $cd^{-1}$ & & & & Y \\
$K_3$  & $cd^{-1/2}$ &  \cite{SHE92}& & & Y \\
$K_4$  & & & $cd^{-2/3}$  & Prop~\ref{thm:k4-lb} &  \\
$K_r$ & & & $cd^{-1+1/(2r-4)}$ & Thm~\ref{thm:kr-lb} &  \\
 & & & $cd^{-5/7 + \varepsilon_r}$ & Prop~\ref{thm:conditional-kr-lb} if Conj~\ref{conj:abks} &  \\
$C_r$ &   & & $cd^{-1/2}$ & Prop~\ref{thm:generalh-lb} & Y \\
$W_r$ for odd $r$ & &  & $cd^{-(r+1)/(2r-1)}$ & Prop~\ref{thm:generalh-lb} &  \\
$K_{2,s}$  & &  & $cd^{-1/2}$  & Prop~\ref{thm:generalh-lb} & Y \\
$K_{3,s}$  & &  & $cd^{-1/2}$ & Prop~\ref{prop:k3s} & Y \\
$K_{4,s}$  & &  & $cd^{-2/3}$  & Prop~\ref{thm:generalh-lb} &  \\
forest   & $\frac{1}{2r}$ & \cite{AL05} & &  & Y \\
forest+1   & &  & $cd^{-1/2}$  & Prop~\ref{thm:generalh-lb} & Y \\
forest+2   & &  & $cd^{-2/3}$  & Prop~\ref{thm:generalh-lb} & \\
\end{tabular}
\caption{Lower bounds for the surplus of $H$-free $d$-degenerate graphs in the literature and our work. They are noted as tight if there is a known construction whose surplus is within a constant factor of the lower bound.}
\label{t:lb}
\end{center}
\end{table}

\section{A stronger conjecture}
\label{sec:conj}

Here, we do some technical work to show that Conjecture~\ref{conj:opt} implies the more well known Conjecture~\ref{conj:abks}.
Again, Conjecture~\ref{conj:opt} easily implies that the surplus for an $H$-free graph is $\Omega(m^{3/4})$, and here we show that the surplus is in fact $\Omega(m^{3/4+\varepsilon})$ for some $\varepsilon=\varepsilon(H)>0$.

\begin{theorem}
  Conjecture~\ref{conj:opt} implies Conjecture~\ref{conj:abks}.
\label{thm:conj}
\end{theorem} 

For a graph $G$ and a subset $U$ of the vertices, let $m(U)$ denote the number of edges in the induced subgraph $G[U]$.
We first observe that large cuts in induced subgraphs can be extended to large cuts in the overall graph.
\begin{lemma}
  \label{lem:conj-add}
  Let $G$ be a graph and $U$ be a subset of the vertices.
  If the induced subgraph $G[U]$ has a cut of size at least $\frac{m(U)}{2} + C$ for some $C > 0$, then $f(G) \ge \frac{m}{2}+C$.
\end{lemma}
\begin{proof}
Fix a cut of $G[U]$ into vertex sets $U_1 \sqcup U_2 = U$ of size at least $m(U)/2$. Then, for all $v \in V \backslash U$, uniformly at random add $v$ to either $U_1$ or $U_2$ (cutting any internal edges) to grow $U_1 \sqcup U_2$ into a partition of $V$ that induces a cut of expected size at least
$$\frac{m - m(U)}{2} + \frac{m(U)}{2} + C = \frac{m}{2} + C.$$
Thus, there exists a cut of $G$ with at least this size, as desired.
\end{proof}

In the next lemma, we show that a graph with few $K_{r+1}$'s and with every vertex participating in many $K_r$'s has a cut with large advantage over a random cut.
To do this, we adapt an argument of \cite{AL96} to show that such a graph has a large subgraph with small chromatic number.
Hence, this large subgraph has a cut with a significant advantage over a random cut.
This cut can then be extended (using Lemma~\ref{lem:conj-add}) to a cut over the original graph with large advantage.

\begin{lemma}
  \label{lem:conj-color}
  Let $r$ be an integer at least $2$.
  For any $\delta\in(0,1)$, Then, for all graphs $G=(V,E)$ on $n$ vertices and $m$ edges with $n$ sufficiently large, if $G$ contains at most $n^{r+1-\delta}$ copies of $K_{r+1}$ and each $v \in V$ is part of at least $n^{r-1-(\delta/3r)}$ many copies of $K_r$, then
  \begin{align}
    f(G) \ge \frac{m}{2} + m^{1-\delta/3}.
  \end{align}
\end{lemma}
\begin{proof}
Let $G = (V, E)$ be as above and let $\varepsilon=\delta/3r$.
Since each $v \in V$ is part of at least $n^{r-1-\varepsilon}$ many copies of $K_r$, the graph $G$ has at least $\frac{1}{r}n^{r-\varepsilon}$ copies of $K_r$.
Since each edge is in at most $n^{r-2}$ many copies of $K_r$, we have $$ m \ge \binom{r}{2}\cdot \frac{1}{r} \cdot \frac{n^{r-\varepsilon}}{n^{r-2}} > \frac{n^{2-\varepsilon}}{2}.$$
  Let $t=64n^{\varepsilon}$, so $m > n^2/t$ and choose a set $T$ of exactly $t$ distinct vertices of $V$ uniformly at random.
  Let $X \subset V$ be the set of vertices that, along with some collection of $r-1$ elements of $T$, form a copy of $K_r$ in $G$.

  We next show that we expect most vertices to lie in $X$.
  Fix some vertex $v \in V$.
  Let $A_1,\dots,A_\ell$ denote the subsets of $r-1$ vertices that form a $K_r$ with $v$, where $\ell \ge n^{r-1-\varepsilon}$.
  For $i=1,\dots,\ell$, let $Z_i$ be the indicator random variable $\mathbf{1}\{A_i \subseteq T\}$. Let random variable $Z := Z_1+\cdots+Z_\ell$.
  Note that
  $$\mathbf{P}(Z_i = 1) = \mathbf{P}(A_i \subset T) = \frac{\binom{n-(r-1)}{t-(r-1)}}{\binom{n}{t}} \ge \frac{t^{r-1}}{2n^{r-1}},$$
  where the inequality holds if $n$ is sufficiently large.
 Thus, $$\E[Z] = \sum_{i = 1}^{\ell} \E[Z_i] \ge n^{r-1-\varepsilon}\cdot \frac{t^{r-1}}{2n^{r-1}} = \frac{t^{r-1}}{2n^{\varepsilon}}.$$
  If $A_i$ and $A_j$ are disjoint,  $Z_i$ and $Z_j$ are negatively correlated, so $\E[Z_iZ_j]-\E[Z_i]\E[Z_j]\le 0$.
  If $|A_i\cup A_j|=s$ for $r\le s\le 2r-3$, then we have $\E[Z_iZ_j] = \frac{\binom{n-s}{t-s}}{\binom{n}{t}}\le \frac{t^s}{n^s}$.
  Furthermore, for $r\le s\le 2r-3$, there are at most $n^{r-1-\varepsilon}\cdot n^{s-(r-1)}= n^{s-\varepsilon}$ pairs $(A_i,A_j)$ such that $|A_i\cap A_j|=s$.
  Thus,
  \begin{align}
    \Var[Z]
    \ &= \   \sum_{i,j}^{} \E[Z_iZ_j] - \E[Z_i]\E[Z_j]
    \ \le \ \sum_{s=r}^{2r-3} \sum_{i,j:|A_i\cap A_j|=s}^{} \E[Z_iZ_j]  \nonumber\\
    \ &\le \ \sum_{s=r}^{2r-3} \frac{t^s}{n^s} \cdot \#\{i,j: |A_i\cup A_j|=s\}
    \ \le \ \sum_{s=r}^{2r-3} t^sn^{-\varepsilon}
    \ < \ 2t^{2r-3}n^{-\varepsilon}.
  \end{align}
  For all random variables, we have $\Pr[Z=0]\le \frac{\Var[Z]}{\E[Z]^2}$ (see, e.g. Theorem~4.3.1 of \cite{AS92}).
  Hence,
  \begin{align}
    \Pr[v\notin X]
    \ = \  \Pr[Z=0]
    \ \le \ \frac{\Var[Z]}{\E[Z]^2}
    \ < \ \frac{2t^{2r-3}/n^{\varepsilon}}{(t^{r-1}/2n^{\varepsilon})^2}
    \ = \  \frac{1}{8}.
  \end{align}
  Thus, the probability an edge has at least one vertex not in $X$ is less than $\frac{1}{4}$, so the expected number of edges not in $X$ is less than $\frac{m}{4}$.
  Thus, by Markov's inequality, with probability less than $\frac{1}{2}$, at most $\frac{m}{2}$ edges are in $X$.

  Call an $(r+1)$-clique of $G$ \textit{bad} if exactly $r-1$ of the vertices are in $T$.
  Each $(r+1)$-clique is bad with probability at most $\binom{r+1}{r-1}\frac{\binom{t-(r-1)}{n-(r-1)}}{\binom{n}{t}}< \frac{r^2t^{r-1}}{n^{r-1}}$.
  As there are at most $n^{r+1-\delta}$ many $(r+1)$-cliques, the expected number of bad cliques is at most $r^2t^{r-1}n^{2- \delta}$.
  By Markov's inequality, with probability at least 1/2, there are at most $2r^2t^{r-1}n^{2-\delta}$ bad cliques.
  This means that there exists some subset $T$ of $t$ vertices such that (1) the corresponding $X$ has $m(X)\ge \frac{m}{2}$ edges and (2) there are at most $2r^2t^{r-1}n^{2-\delta}$ bad cliques.

  Fix this $T$, and let $G'$ be the graph on vertex set $X$ obtained by removing the edges from every bad $(r+1)$-clique in the induced subgraph $G[X]$.
  The total number of edges in bad cliques is at most
  \begin{align}
    \binom{r+1}{2}\cdot 2r^2t^{r-1}n^{2-\delta}
    \ \overset{(*)}< \ \frac{n^2}{t^{r+1}}
    \ \overset{(**)}< \ \frac{m}{2t^{r}}
    \ \le \ t^{-r}m(X).
  \end{align}
  In $(*)$, we used that $2r^4t^{2r} = c_r n^{2r\varepsilon} <n^{\delta}$ for $n$ sufficiently large.
  In $(**)$, we used that $m > n^2/t$.
  Hence, $G'$ has at least $m(X)\cdot (1 - t^{-r})$ edges.
  Additionally, $\chi(G')\le \binom{t}{r-1}$, seen by coloring each vertex $v\in X$ with an unordered $(r-1)$-tuple corresponding to a subset of $(r-1)$ vertices in $T$ that form a $K_r$ with $v$.
  By definition of $X$, such an $(r-1)$-tuple exists.
  Since $G'$ has no edge forming a $K_{r+1}$ with $r-1$ elements of $T$, the above coloring is a proper coloring of $X$.
  Hence, by Lemma~\ref{lem:mcchi},
  \begin{align}
    f(G')
    \ &\ge \ \left(\frac{1}{2} + \frac{1}{2\binom{t}{r-1}}\right)\cdot m(X)\cdot(1-t^{-r})  \nonumber\\
    \ &> \ \left(\frac{1}{2} + \frac{1}{4\binom{t}{r-1}}\right)m(X)
    \ \overset{(*)}> \ \frac{m(X)}{2} + m^{1-\delta/3},
  \end{align}
  where $(*)$ follows since $\frac{m(X)}{4\binom{t}{r-1}}\ge \frac{m}{8t^{r-1}} > \frac{m}{n^{r\varepsilon}} > m^{1-r\varepsilon} = m^{1-\delta/3}$.
  Hence, the induced subgraph $G[X]$ has a cut of at least the same value.
  By Lemma~\ref{lem:conj-add}, $G$ has a cut of size $\frac{m}{2} + m^{1-\delta/3}$.
\end{proof}

In the next lemma, we show that a graph with few $K_{r+1}$'s and many edges has a cut with large advantage over a random cut.
To do this, we induct on $r$.
We show there are two nontrivial cases: either (1) there is a subgraph with many edges and few $K_r$'s, in which case we apply the induction hypothesis or (2) there is some subgraph with many edges and every vertex is in many $K_r$'s, in which case we apply Lemma~\ref{lem:conj-color}.

\begin{lemma}
  \label{lem:conj-induct}
  Let $r\ge 1$.
  Let $\delta\in(0,1)$.
  For $n$ sufficiently large, every graph $G$ on $n$ vertices with more than $n^{2-\delta/(2^rr!)}$ edges and at most $n^{r+1-\delta}$ many $K_{r+1}$'s, has $f(G) \ge \frac{m}{2} + m^{1-\delta}$.
\end{lemma}
\begin{proof}
  We prove by induction on $r$.
  For $r=1$, the statement is vacuous: no graph $G$ has more than $n^{2-\delta/2}$ edges while also having at most $n^{2-\delta}$ many $K_2$'s.

  Assume the assertion is true for $r-1$.
  For simplicity, let $\varepsilon=\frac{\delta}{2^rr!}$.
  Let $\delta' = \frac{11\delta}{20r}$ and $\varepsilon' = \frac{\delta'}{2^{r-1}(r-1)!}$, so that $\varepsilon'  > \varepsilon$.

  Suppose $G$ has at most $n^{r+1-\delta}$ many $K_{r+1}$'s and $m\ge n^{2-\varepsilon}$ edges.
  Suppose we find a vertex of $G$ contained in less than $n^{r-1-(\delta/6r)}$ many $K_r$'s, delete it, and repeat on the resulting graph until no such vertex exists.
  Let $W$ be the set of vertices that remain after this procedure, and let $U$ be the set of vertices that are deleted.
  We have three cases.
  \begin{enumerate}

\item[Case 1] (Easy). If there are at least $\frac{2m}{3}$ edges between $U$ and $W$, then $(U,W)$ forms a cut of $G$ with at least $\frac{2m}{3} > \frac{m}{2} + m^{1-\delta}$ edges.

  \item[Case 2] (Few $K_r$'s). If there are at least $\frac{m}{6}$ edges in the induced subgraph $G[U]$, then the following two statements are true about $G[U]$:
  \begin{itemize}
  \item \textit{The induced subgraph $G[U]$ has at most $|U|^{r-\delta'}$ many $K_r$'s.}

  Since $G[U]$ has at least $m/6$ edges, $$|U| \ge \sqrt{\frac{m}{3}} > \frac{n^{1-\varepsilon/2}}{2}.$$
  When each vertex in $U$ was deleted, it was in at most $n^{r-1-(\delta/6r)}$ many $K_r$'s. Thus, the total number of $K_r$'s of $G$ that touch the vertex subset $U$ is at most $$|U|n^{r-1-(\delta/6r)}<n^{r-(\delta/6r)}.$$
  Hence, $G[U]$ has at most $n^{r-(\delta/6r)} \le |U|^{r-\delta'}$ many $K_r$'s; the inequality follows since $n$ is sufficiently large and $r-\delta' <  (r-\frac{\delta}{6r})(1-\frac{\varepsilon}{2}).$

  \item \textit{The induced subgraph $G[U]$ has at least $|U|^{2-\varepsilon'}$ edges.}
 This follows since
  $$m(G[U]) = \frac{n^{2-\varepsilon}}{6} \ge \frac{|U|^{2-\varepsilon}}{6} \ge |U|^{2-\varepsilon'},$$
 which holds since $\varepsilon'>\varepsilon$ and $n$ is sufficiently large.
  \end{itemize}
  By the above two properties, the $G[U]$ satisfies the setup of the inductive hypothesis, with parameters $r-1$ and $\delta'$.
  Hence, by the inductive hypothesis, we have that for sufficiently large $n$
  \begin{align}
    f(G[U])\ge \frac{m(U)}{2} + (m/6)^{1-\delta'} > \frac{m(U)}{2} + m^{1-\delta}
    \ge \frac{m}{2} + m^{1-\delta},
  \end{align}
  since $m(U)>m/6$ and $\delta' < \delta$, applying Lemma~\ref{lem:conj-add}.

  \item[Case 3] (Many $K_r$'s). If there are at least $\frac{m}{6}$ edges in the induced subgraph $G[W]$, the following two statements are true about the induced subgraph $G[W]$.

  \begin{itemize}
  \item \textit{Each vertex is in at least $|W|^{r-1-\delta/6r}$ many $K_r$'s.}

  By construction, each vertex is in at least $n^{r-1-\delta/6r}$ many $K_r$'s, or else we would have deleted it in the above procedure.
  Furthermore $n\ge |W|$, so each vertex is in at least $|W|^{r-1-\delta/6r}$ many $K_r$'s.

  \item \textit{It has at most $|W|^{r+1-\delta/2}$ many $K_r$'s.}

  Since $G[W]$ has at least $\frac{m}{6}$ edges, $W$ has at least $\sqrt{\frac{m}{3}} > \frac{n^{1-\varepsilon/2}}{2}$ vertices.
  In $G[W]$, there are at most $n^{r+1-\delta} \le |W|^{r+1 - \delta/2}$ many $K_{r+1}$'s, which holds since $$r+1-\delta < \left(1-\frac{\varepsilon}{2}\right)\left(r+1-\frac{\delta}{2}\right).$$
  \end{itemize}
  By the above two properties, $G[W]$ satisfies the setup of the Lemma~\ref{lem:conj-color} with parameters $r$ and $\frac{\delta}{2}$.
  Hence, by Lemma~\ref{lem:conj-color}, we have for sufficiently large $m$
  \begin{align}
    f(G[W])\ge \frac{m(W)}{2} + (m/6)^{1-\delta/2} > \frac{m(W)}{2} + m^{1-\delta}.
  \end{align}
  Hence, by Lemma~\ref{lem:conj-add}, we have
  \begin{align}
    f(G)\ge \frac{m}{2} + m^{1-\delta}.
  \end{align}
\end{enumerate}
  This covers all the cases, and in each case, we have $f(G)\ge \frac{m}{2} + m^{1-\delta}$, as desired.
\end{proof}

The above tools will enable us to show Theorem~\ref{thm:conj}.

\begin{proof}[Proof of Theorem~\ref{thm:conj}]
  Fix $r\ge 2$.
  Assume Conjecture~\ref{conj:opt} is true.
  It suffices to lower bound the Max-Cut of $K_{r+1}$-free graphs, since all graphs are subgraphs of a clique.
  Let $\delta=\frac{1}{5}$ and $\varepsilon=\frac{\delta}{2^rr!}$.
  Suppose $G$ is a graph with $m$ edges and $n$ vertices.
  We show that $G$ has a cut of size $\frac{m}{2} + \Omega(m^{3/4+\varepsilon/8})$ in two cases:
  Let $d=m^{1/2 - \varepsilon/4}$ and assume $m$ and $n$ are sufficiently large.
  \begin{enumerate}
  \item[Case 1] (Sparse: $G$ has no induced subgraph of minimum degree $d$).
  This implies that $G$ is $d$-degenerate, in which case Conjecture~\ref{conj:opt} implies that for some $c > 0$ $$f(G) \ge \frac{m}{2} + \frac{cm}{\sqrt{d}} = \frac{m}{2} + cm^{3/4 + \varepsilon/8}.$$

  \item[Case 2] (Dense: there exists an induced subgraph $G[U]$ of minimum degree $d$).
  Then
  \begin{align}
    m(U)\ge \frac{|U|d}{2} = \frac{|U|\cdot m^{1/2-\varepsilon/4}}{2} \ge \frac{|U|\cdot m(U)^{1/2-\varepsilon/4}}{2}.
  \end{align}
  Rearranging, and using that $|U|$ and $m(U)$ are sufficiently large, gives that
  \begin{align}
    m(U) > |U|^{2-\varepsilon}.
  \end{align}
  Since $G[U]$ has $0<n^{r+1-\delta}$ many $K_{r+1}$'s, we may apply Lemma~\ref{lem:conj-induct} to obtain that $G[U]$ has a cut of size
  \begin{align}
     f(G[U])
    \ \ge \  \frac{m(U)}{2} + m(U)^{1-\delta}
  \end{align}
  We know that $m(U)\ge \frac{d^2}{2} \ge \frac{m^{1-\varepsilon/2}}{2}$, so
  \begin{align}
     f(G[U])
    \ > \ \frac{m(U)}{2} + \frac{m^{(1-\varepsilon/2)(1-\delta)}}{2^{1-\delta}}
    \ > \ \frac{m(U)}{2} + m^{3/4 + \varepsilon/8}.
  \end{align}
  In the last inequality, we used that $\delta=\frac{1}{5}$ and $\varepsilon \le \frac{\delta}{8}=\frac{1}{40}$.
  By Lemma~\ref{lem:conj-add}, we have
  \begin{align}
    f(G)
    \ \ge \ \frac{m}{2} + m^{3/4+\varepsilon/8}.
  \end{align}
\end{enumerate}
  This completes the proof.
\end{proof}
\begin{remark}
  The above argument proves that, assuming Conjecture~\ref{conj:opt}, a $K_r$-free graph with $m$ edges has Max-Cut value at least $\frac{m}{2} + c_rm^{3/4 + \varepsilon_r}$ for $\varepsilon_r = 2^{-\Theta(r\log r)}$.
  For clarity, we did not optimize the value of $\varepsilon_r$.
\end{remark}

\section{Max-Cut upper bound matching Conjecture~\ref{conj:opt}}
\label{app:ub}

Alon, Krivelevich, and Sudakov \cite{AL05} showed that, when $H$ is a forest, the $\MC$ of $H$-free graphs is $(\frac{1}{2}+c)\cdot m$ for some $c > 0$ independent of $m$.
This result holds independently of the density of the graph, and in particular also applies to $d$-degenerate graphs, where the constant in the lower order term is independent of $d$.

For $d$-degenerate graphs, we observe that forests are the \emph{only} graphs for which this is true: whenever $H$ contains a cycle, there exist infinitely many $H$-free $d$-degenerate (and, in fact, maximum degree $d$) graphs $G$ on $n$ vertices with $\MC$ no larger than $(\frac{1}{2} + \frac{c}{\sqrt{d}})\cdot m(G)$.
In particular, Conjecture~\ref{conj:opt} is optimal (up to a constant depending on $H$ in the lower order term) if it is true when $H$ is not a forest.
\begin{prop}
\label{prop:ub}
  For all $r\ge 3$ and $d\ge 1$, there exist $c=c(r)>0$ and $n_0=n_0(r,d)$ such that for all $n\ge n_0$, there exists a $C_r$-free graph $G$ on $n$ vertices with maximum degree $d$ and
  \begin{align}
    f(G)\le \left(\frac{1}{2} + \frac{c}{\sqrt{d}}\right)\cdot m(G).
  \end{align}
\end{prop}

Let $G_{n, d}$ denote a random $d$-regular graph on $n$ vertices. We will show that a random regular graph $G_{n,d}$ (with a few alterations to make it $C_r$-free) gives the desired bound.
The following result of~\Bollobas~\cite{Bol80} implies that a random regular graph has few $r$-cycles with high probability.
\begin{prop}[Theorem 2 of \cite{Bol80}]
  \label{prop:bol}
  For $r\ge 3$ and $d\ge 1$ fixed, as $n\to\infty$, the distribution of the number of copies of $C_r$ in a random $d$-regular graph $G_{n,d}$ converges to $Poi(\lambda)$ for $\lambda=(d-1)^r/2r$.
\end{prop}
The following result (e.g. in \cite{DMS17}) shows that with high probability, random regular graphs have $\MC$ within a constant factor of the bound in Conjecture~\ref{conj:opt}.

\begin{prop}[\cite{DMS17}]
  \label{prop:dms}
  There exists an absolute constant $c>0$ such that, for any $d \ge 1$, there exists $n_d$ such that for all $n\ge n_d$, with probability at least $0.99$, the random regular graph $G_{n,d}$ has $\MC$ at most $(\frac{1}{2} + \frac{c}{\sqrt{d}})m$, where $m=\frac{dn}{2}$.
\end{prop}
Combining the above two results gives Proposition~\ref{prop:ub}.
\begin{proof}[Proof of Proposition~\ref{prop:ub}]
  Let $\lambda = \frac{(d-1)^r}{2r}$.
  By Proposition~\ref{prop:bol}, there exists an $n_0'$ such that, for all $n\ge n_0'$ the probability that a random regular graph $G_{n,d}$ has at least $d^r > 6\lambda$ copies of $C_r$ is at most $e^{-6} < 0.01$.
  Let $c>0$ and $n_d$ be given by Proposition~\ref{prop:dms}, and let $n_0=\max(n_0',n_d, 2d^r)$. For all $n\ge n_0$, with probability at least $0.98$, a random regular graph $G_{n,d}$ has at most $d^r$ copies of $C_r$ and $\MC$ at most $(\frac{1}{2} + \frac{c}{\sqrt{d}})\cdot\frac{dn}{2}$.
  Let $G'$ be such a graph, and let $G$ be the graph obtained by removing (at least) one edge from each $C_r$, so that at most $d^r$ edges are removed, and $G$ has $m(G)\ge \frac{dn}{2} - d^r \ge \frac{dn}{2}(1-\frac{1}{2d})$ edges.
  Then, the $\MC$ of $G$ is at most
  \begin{align}
    g(G)
    \ &\le \  \left( \frac{1}{2} + \frac{c}{\sqrt{d}} \right)\cdot \frac{dn}{2}
    \ &\le \ \left( \frac{1}{2} +  \frac{c}{\sqrt{d}}\right)\frac{m(G)}{1-(1/2d)}
    \ < \   \left( \frac{1}{2} + \frac{c'}{\sqrt{d}} \right)m(G)
  \label{}
  \end{align}
  for some $c'>0$.
\end{proof}

\end{document}